\documentclass[a4paper]{article}
\usepackage{a4wide}
\usepackage{graphicx}
\usepackage{amsmath}
\usepackage{amsthm}
\usepackage{amssymb, latexsym, mathrsfs}
\usepackage{enumerate}
\usepackage{algorithm}
\usepackage{subfigure}
\usepackage{color}
\usepackage{transparent}
\usepackage{dsfont}
\usepackage{url}
\usepackage[affil-it]{authblk}
\theoremstyle{plain}

\newtheorem{lem}{Lemma}

\newtheorem{prob}{Problem}

\newtheorem{rmk}{Remark}
\newtheorem{prop}{Proposition}

\def\Rset{\mathbb{R}}
\def\Nset{\mathbb{N}}

\newcommand{\FF}{{\mathcal F}}

\newcommand{\HH}{{\mathcal H}}

\newcommand{\MM}{{\mathcal M}}
\newcommand{\NN}{{\mathcal N}}

\newcommand{\PP}{{\mathcal P}}

\newcommand{\SSS}{{\mathcal S}}

\newcommand{\mbf}[1]{\mathbf{#1}}
\newcommand{\px}{{x^+}}
\newcommand{\subss}[2]{#1_{[#2]}}
\newcommand{\tx}{{\tilde x}}

\newcommand{\diag}{\mbox{diag}}
\newcommand{\rank}{\mbox{rank}}
\newcommand{\Xset}{\mathbb{X}}

\newcommand{\Wset}{\mathbb{{{W}}}}

\newcommand{\Vset}{\mathbb{V}}
\newcommand{\ball}[1]{{B_{#1}}}
\newcommand{\abs}[1]{{|{#1}|}}
\newcommand{\norme}[2]{{||{#1}||_{#2}}}

\newcommand{\One}{\textbf{1}}
\newcommand{\Zero}{\textbf{0}}
\newcommand{\Metz}{\mathds{M}}

\newcommand{\eye}[1]{\mathds{I}_#1}

\newcommand{\tpx}{{\tilde{x}}^+}
\newcommand{\pe}{{e}^+}
\newcommand{\Eset}{\mathbb{E}}
\newcommand{\Sset}{\mathbb{S}}

\newcommand{\oneblock}[1]{
\begin{aligned}
#1
\end{aligned}
}

\newcommand{\imply}{\Rightarrow}

\newcommand{\ba}[1]{\begin{array}{#1}}
\newcommand{\ea}{\end{array}}

\newcommand{\matr}[1]{
\begin{bmatrix}
    #1
\end{bmatrix}
}

\begin{document}

      \title{Plug-and-play distributed state estimation for linear systems\thanks{The research leading to these results has received funding from the European Union Seventh Framework Programme [FP7/2007-2013]  under grant agreement n$^\circ$ 257462 HYCON2 Network of excellence.}}
     \author{Stefano Riverso%
       \thanks{Electronic address: \texttt{stefano.riverso@unipv.it}; Corresponding author}} 
     
     \author{Marcello Farina%
       \thanks{Electronic address: \texttt{farina@elet.polimi.it}}} 

     \author{Riccardo Scattolini%
       \thanks{Electronic address: \texttt{riccardo.scattolini@elet.polimi.it}}} 

     \author{Giancarlo Ferrari-Trecate%
       \thanks{Electronic address: \texttt{giancarlo.ferrari@unipv.it}\\S. Riverso and G. Ferrari-Trecate are with Dipartimento di Ingegneria Industriale e dell'Informazione, Universit\`a degli Studi di Pavia, via Ferrata 1, 27100 Pavia, Italy\\M. Farina and R. Scattolini are with Dipartimento di Elettronica e Informazione, Politecnico di Milano, via Ponzio 34/5, 20133 Milan, Italy}} 

     \affil{Dipartimento di Ingegneria Industriale e dell'Informazione\\Universit\`a degli Studi di Pavia}
     \date{\textbf{Technical Report}\\ September, 2013}

     \maketitle

     \begin{abstract}
       This paper proposes a state estimator for large-scale linear systems described by the interaction of state-coupled subsystems affected by bounded disturbances. We equip each subsystem with a Local State Estimator (LSE) for the reconstruction of the subsystem states using pieces of information from parent subsystems only. Moreover we provide conditions guaranteeing that the estimation errors are confined into prescribed polyhedral sets and converge to zero in absence of disturbances. Quite remarkably, the design of an LSE is recast into an optimization problem that requires data from the corresponding subsystem and its parents only. This allows one to synthesize LSEs in a Plug-and-Play (PnP) fashion, i.e. when a subsystem gets added, the update of the whole estimator requires at most the design of an LSE for the subsystem and its parents. Theoretical results are backed up by numerical experiments on a mechanical system.
     \end{abstract}

     \newpage

     \section{Introduction}
          In several applications, the use of centralized state estimators is hampered by the complexity of the underlying systems. As an example, when plants are composed by several subsystems arranged in a parent-child coupling relation, online operations, such as the transmission of output samples to a central processing unit or the simultaneous estimation of all states, can be prohibitive. This has motivated a large body of research on Distributed State Estimators (DSEs) where subsystems are equipped with LSEs connected through a communication network and dedicated to the reconstruction of local states only \cite{Mutambara1998,Vadigepalli2003,Khan2008,Stankovic2009,Stankovic2009a,Farina2010,Farina2011b,Riverso2013b}. Concerning the required communication links, some methods are more parsimonious as they do not need information to be exchanged between all LSEs, but only along the edges of a directed network with the parent-child topology induced by subsystems coupling \cite{Khan2008,Stankovic2009,Stankovic2009a,Farina2010,Farina2011b,Riverso2013b}. Furthermore, there are methods that also guarantee the fulfillment of constraints on local states \cite{Farina2010} or estimation errors \cite{Farina2011b,Riverso2013b}.

          As in  \cite{Farina2011b} and \cite{Riverso2013b}, in this paper we consider discrete-time linear time-invariant subsystems affected by bounded disturbances and propose a DSE composed by LSEs with a Luenberger-like structure and connected through a network with parent-child topology. We provide conditions for guaranteeing estimation errors fulfill prescribed polyhedral constraints at all times and converge to zero when there are no disturbances. A key feature of our approach is that, differently from \cite{Farina2011b} and \cite{Riverso2013b}, checking these conditions amounts to numerical tests that are associated with individual LSEs and that can be conducted in parallel using hardware collocated with subsystems. Furthermore, each test requires data from parent subsystems only. These properties enable PnP design of LSEs, meaning that (i) when a subsystem is added to a plant, the corresponding LSE can be designed using pieces of information from parent subsystems only; (ii) in order to preserve the key
properties of the whole DSE, the plugging in and out of a subsystem triggers at most the update of LSEs associated to child subsystems and (iii) the design/update of an LSE is automatized, e.g. it is recast into an optimization problem that can be solved using local hardware. We highlight that addition and removal of subsystems, as well as synthesis of LSEs, are here considered as offline operations and therefore no hybrid dynamics is generated. Our method, that parallels the PnP procedure for the design of decentralized model predictive controllers proposed in \cite{Riverso2013c} and \cite{Riverso2012h}, can be useful in the context of systems of systems \cite{Samad2011} and cyber-physical systems \cite{Antsaklis2013} where, typically, the number of subsystems changes over time.

          The paper is structured as follows. The DSE is introduced in Section \ref{sec:distrstateesit}. In Section \ref{sec:dec_design}, the main results allowing design decentralization are presented together with the optimization-based synthesis of LSEs. PnP operations are discussed in \ref{sec:plugplay}. In Section \ref{sec:example} we illustrate the use of the DSE for reconstructing the states of a 2D array of masses connected by springs and dampers. Finally, Section \ref{sec:conclusions} is devoted to conclusions.

          \textbf{Notation.} We use $a:b$ for the set of integers $\{a,a+1,\ldots,b\}$. The symbol $\Rset_+^n$ stands for the vectors in $\Rset^n$ with nonnegative elements. The column vector with $s$ components $v_1,\dots,v_s$ is $\mbf v=(v_1,\dots,v_s)$. The symbol $\oplus$ denotes the Minkowski sum, i.e. $A=B\oplus C$ if and only if $A=\{a:a=b+c,~b\in B, ~c\in C\}$. Moreover, $\bigoplus_{i=1}^sG_i=G_1\oplus\ldots\oplus G_s$. The symbol $\One_\alpha$ (resp. $\Zero_\alpha$) denotes a column vector with $\alpha\in\Nset$ elements all equal to $1$ (resp. $0$). Given a matrix $A\in\Rset^{n\times n}$, with entries $a_{ij}$ its entry-wise 1-norm is  $\norme{A}{1}=\sum_{i=1}^n\sum_{j=1}^n\abs{a_{ij}}$ and its Frobenius norm is $\norme{A}{F}=\sqrt{\sum_{i=1}^n\sum_{j=1}^na_{ij}^2}$. The standard Euclidean norm is denoted with $\| \cdot\|$. The pseudo-inverse of a matrix $A\in\Rset^{m\times n}$ is denoted with $A^\flat$.\\
          The set $\Xset\subseteq\Rset^n$ is positively invariant \cite{Rawlings2009} for $x(t+1)=f(x(t))$, if $x(t)\in\Xset\imply f(x(t))\in\Xset$.\\
          The set $\Xset\subseteq\Rset^n$ is Robust Positively Invariant (RPI) \cite{Rawlings2009} for $x(t+1)=f(x(t),w(t))$, $w(t)\in\Wset\subseteq\Rset^m$ if $x(t)\in\Xset\imply f(x(t),w(t))\in\Xset\mbox{, }\forall w(t)\in\Wset$.  The RPI set $\bar\Xset$ is maximal (MRPI) if every other RPI $\Xset$ verifies $\bar{\Xset}\supseteq\Xset$. The RPI set $\underline\Xset$ is minimal (mRPI) if every other RPI $\Xset$ verifies $\underline{\Xset}\subseteq\Xset$. The RPI set $\Xset(\epsilon)$ is a $\epsilon$-outer approximation of the mRPI $\underline\Xset$ if
          \begin{equation*}
            \label{eq:defapproxmRPI}
            x\in\Xset(\epsilon)\imply\exists~\underline x\in\underline\Xset\mbox{ and } \tx\in\ball{\epsilon}(0): x = \underline x + \tx
          \end{equation*}
          where, for $\epsilon>0$,  $\ball{\epsilon}(v)=\{x\in\Rset^n|\|x-v\|<\epsilon\}$.

     \section{Distributed state estimator}
          \label{sec:distrstateesit}
          We consider a discrete-time Linear Time Invariant (LTI) system
          \begin{equation}
            \label{eq:model}
            \begin{aligned}
              \mbf\px &= \mbf{Ax+Bu+Dw}\\
              \mbf y &= \mbf{Cx}
            \end{aligned}
          \end{equation}
          where $\mbf x\in\Rset^n$, $\mbf u\in\Rset^m$, $\mbf y\in\Rset^p$ and $\mbf w\in\Rset^r$ are the state, the input, the output and the disturbance, respectively, at time $t$ and $\mbf\px$ stands for $\mbf x$ at time $t+1$. The state is composed by $M$ state vectors $\subss x i\in\Rset^{n_i}$, $i\in\MM=1:M$ such that $\mbf x=(\subss x 1,\ldots,\subss x M)$, and $n=\sum_{i\in\MM}n_i$. Similarly, the input, the output and the disturbance are composed by $M$ vectors $\subss u i\in\Rset^{m_i}$,  $\subss y i\in\Rset^{p_i}$,  $\subss w i\in\Rset^{r_i}$, $i\in\MM$ such that $\mbf u=(\subss u 1,\ldots,\subss u M)$, $m=\sum_{i\in\MM}m_i$, $\mbf y=(\subss y 1,\ldots,\subss y M)$, $p=\sum_{i\in\MM}p_i$, $\mbf w=(\subss w 1,\ldots,\subss w M)$ and $r=\sum_{i\in\MM}r_i$.
          
          We assume \eqref{eq:model} can be equivalently described by subsystems $\subss \Sigma i$, $i\in\MM$, given by
          \begin{equation}
            \label{eq:subsystem}
            \begin{aligned}
              \subss\Sigma i:\quad\subss \px i&=A_{ii}\subss x i+B_i\subss u i+\sum_{j\in\NN_i}A_{ij}\subss x j+D_i\subss w i\\
              \subss y i&=C_{i}\subss x i
            \end{aligned}
          \end{equation}
          where $A_{ij}\in\Rset^{n_i\times n_j}$, $i,j\in\MM$, $B_i\in\Rset^{n_i\times m_i}$, $D_i\in\Rset^{n_i\times r_i}$, $C_i\in\Rset^{p_i\times n_i}$ and $\NN_i$ is the set of parents of subsystem $i$ defined as $\NN_i=\{j\in\MM:A_{ij}\neq 0, i\neq j\}$. Moreover, since $\subss y i$ depends on the local state $\subss x i$ only, subsystems $\subss\Sigma i$ are output-decoupled and then $\mbf{C}=\diag(C_1,\ldots,C_M)$. Similarly, subsystems $\subss\Sigma i$ are input- and disturbance-decoupled, i.e.  $\mbf{B}=\diag(B_1,\ldots,B_M)$ and $\mbf{D}=\diag(D_1,\ldots,D_M)$. We also assume
          \begin{equation}
            \label{eq:Wbound}
            \subss w i\in\Wset_i
          \end{equation}
          where the set $\Wset_i\subset\Rset^{r_i}$ is a zonotope centered at the origin, i.e.  a polytope that is centrally symmetric about the origin. Without loss of generality, $\Wset_i$ can be written as
          \begin{equation}
            \label{eq:setspolyW}
            \begin{aligned}
              \Wset_i %&= \{\subss{w}{i}\in\Rset^{r_i}|f_{i,\upsilon}^T\subss{w}{i}\leq 1,\forall \upsilon\in 1:\bar{\upsilon}_i\} \\
              &= \{\subss{w}{i}\in\Rset^{r_i}|\FF_i\subss{w}{i}\leq \One_{\bar{\upsilon}_i} \} \\
              &=\{\subss{w}{i}\in\Rset^{r_i}|\subss w i=\Delta_i l_i,\mbox{ }\norme{l_i}{\infty}\leq 1 \}
            \end{aligned}
          \end{equation}
          where $\FF_i=(f_{i,1}^T,\ldots,f_{i,\bar \upsilon_i}^T)\in\Rset^{\bar{\upsilon}_i\times r_i}$, $\rank(\FF_i)=r_i$, $\Delta_i\in\Rset^{r_i\times \bar r_i}$ and $l_i\in\Rset^{\bar r_i}$.

          In this section we propose a Distributed State Estimator (DSE) for (\ref{eq:model}). As in \cite{Farina2011b} and \cite{Riverso2013b}, we define for $i\in\MM$ the Local State Estimator (LSE)
          \begin{equation}
            \label{eq:subestimator}
            \begin{aligned}
              \subss{\tilde\Sigma}i:\quad\subss \tpx i =A_{ii}\subss \tx i+B_i\subss u i-L_{ii}(\subss y i-C_i\subss\tx i)+\\\sum_{j\in\NN_i}A_{ij}\subss \tx j-\sum_{j\in\NN_i}\delta_{ij}L_{ij}(\subss y j-C_j\subss\tx j)\\
            \end{aligned}
          \end{equation}
          where $\subss\tx i\in\Rset^{n_i}$ is the state estimate, $L_{ij}\in\Rset^{n_i\times p_j}$ are gain matrices and $\delta_{ij}\in\{0,1\}$. This implies that $\subss{\tilde\Sigma} i$ depends only on local variables ($\subss\tx i$, $\subss u i$ and $\subss y i$) and parents' variables ($\subss\tx j$ and $\subss y j$, $j\in\NN_i$). Binary parameters $\delta_{ij}$, $j\in\NN_i$ can be chosen equal to one for exploiting the knowledge of parents' outputs, or equal to zero for reducing the number of transmitted output samples.

          Defining the state estimation error as
          \begin{equation}
            \label{eq:errori}
            \subss e i=\subss x i-\subss\tx i,
          \end{equation}
          from (\ref{eq:subsystem}), \eqref{eq:subestimator} and (\ref{eq:errori}), we obtain the local error dynamics
          \begin{equation}
            \label{eq:erroridyn}
            \subss \Theta i:\quad   \subss\pe i=\bar{A}_{ii}\subss e i+\sum_{j\in\NN_i}\bar{A}_{ij}\subss e j+D_i\subss w i \\
          \end{equation}
          where $\bar{A}_{ii}=A_{ii}+L_{ii}C_i$ and $\bar{A}_{ij}=A_{ij}+\delta_{ij}L_{ij}C_j$, $i\neq j$. Our main goal is to solve the following problem.
          \begin{prob}
            \label{prob:estimator_properties}
            Design in a decentralized fashion LSEs $\subss{\tilde\Sigma} i$, $i\in\MM$ that
            \begin{description}
            \item[(a)] are nominally convergent, i.e. when $\Wset_i=\{0\}$ it holds
              \begin{equation}
                \label{eq:conv}
                \norme{\subss e i(t)}{}\rightarrow 0 \mbox{ as  }t\rightarrow\infty
              \end{equation}
            \item[(b)] guarantee, for suitable initial conditions
              \begin{equation}
                \label{eq:bounderror}
                \subss e i(t)\in\Eset_i,~\forall t\geq 0
              \end{equation}
              where $\Eset_i\subseteq\Rset^{n_i}$ are zonotopes
              centered at the origin
              given by
              \begin{equation}
                \label{eq:setspolyE}
                \begin{aligned}
                  \Eset_i %&= \{\subss{e}{i}\in\Rset^{n_i}|h_{i,\tau}^T\subss{e}{i}\leq 1,\forall \tau\in 1:\bar{\tau}_i\} \\
                                &= \{\subss{e}{i}\in\Rset^{n_i}|\HH_i\subss{e}{i}\leq \One_{\bar{\tau}_i} \} \\
                                &=\{\subss{e}{i}\in\Rset^{n_i}|\subss e i=\Xi_i d_i,\mbox{ }\norme{d_i}{\infty}\leq 1 \}
                \end{aligned}
              \end{equation}
              In \eqref{eq:setspolyE}, $\HH_i=(h_{i,1}^T,\ldots,h_{i,\bar\tau_i}^T)\in\Rset^{\bar{\tau}_i\times n_i}$, $\rank(\HH_i)=n_i$, $\Xi_i\in\Rset^{n_i\times \bar n_i}$ and $d_i\in\Rset^{\bar n_i}$. $\square$
            \end{description}
          \end{prob}
          Defining the variable $\mbf e=(\subss e 1,\ldots,\subss e M)\in\Rset^n$, from (\ref{eq:erroridyn}) one obtains the collective dynamics of the estimation error
          \begin{equation}
            \label{eq:errordyn}
            \begin{aligned}
              \mbf\pe&=\mbf{\bar{A}}\mbf e+\mbf D\mbf w\\
            \end{aligned}
          \end{equation}
          where the matrix $\mbf{\bar{A}}$ is composed by blocks $\bar{A}_{ij}$, $i,j\in\MM$.

          We equip system \eqref{eq:errordyn} with constraints $\mbf e\in\Eset=\prod_{i\in\MM}\Eset_i$ and $\mbf w\in\Wset=\prod_{i\in\MM}\Wset_i$.
          % \begin{rmk}
          %   We note that, also in a centralized fashion, Problem~\ref{prob:estimator_properties} can be solved if there exist a mRPI $\underline\Sset$ for dynamics \eqref{eq:errordyn} such that $\underline\Sset\subseteq\Eset$. In the next section we solve Problem~\ref{prob:estimator_properties} in a decentralized fashion and obviously our procedure fails if $\underline\Sset\supset\Eset$. Therefore sets $\Eset_i$ must be chosen appropriately.
          % \end{rmk}

          Let $\mbf L$ be the matrix composed by blocks $L_{ij}$, $i,j\in\MM$. From \eqref{eq:errordyn}, if $\mbf L$ is such that $\mbf{\bar A}$ is Schur, then property \eqref{eq:conv} holds. Moreover, if there exists an RPI set $\Sset\subseteq\Eset$ for the constrained system \eqref{eq:errordyn}, then $\mbf{e}(0)\in\Sset$ guarantees property \eqref{eq:bounderror}. We highlight that methods based on Linear Programming (LP) for computing $\Sset$ exist \cite{Rakovic2005a,Rakovic2010}. However the resulting LP problems require the knowledge of the collective model \eqref{eq:model} and therefore they become prohibitive for large-scale systems.

          % Furthermore, $\Sset$ can be computed solving a Linear Programming (LP) problem \cite{Rakovic2005a,Rakovic2010}. However the LP problem includes the collective model \eqref{eq:model} in the constraints and computations become prohibitive for large $n$.

          In absence of coupling between subsystems (i.e. $A_{ij}=0$, $i\neq j$) the error dynamics \eqref{eq:erroridyn} are decoupled as well. Therefore, from \eqref{eq:errordyn}, if $L_{ii}$ are such that matrices $\bar A_{ii}$ are Schur, then  \eqref{eq:conv} holds. Furthermore, if there is an RPI set $\Sset_i\subseteq\Eset_i$ for each local error dynamics, property \eqref{eq:bounderror} can be guaranteed by requiring $\subss e i (0)\in\Sset_i$. Since $\Eset_i$ and $\Wset_i$ are polytopes, using the algorithms in \cite{Rakovic2005a,Rakovic2010} the computation of sets $\Sset_i$, $i\in\MM$ requires the solution of $M$ LP problems that can be solved in parallel using computational resources collocated with subsystems.

          In the next section we propose a method for bridging the gap between the two extreme cases described above, i.e. for designing LSEs in a decentralized fashion even in presence of coupling between subsystems.

     \section{Decentralization of LSE design}
          \label{sec:dec_design}
          In the following, we first solve Problem~\ref{prob:estimator_properties} in the case of $\Wset=\{0\}$ i.e. no disturbances act on subsystems \eqref{eq:model}, and then show how to take disturbances into account.

          When $\Wset=\{0\}$, we need to find matrices $L_{ij}$ $i,j\in\MM$ such that system \eqref{eq:errordyn} is asymptotically stable. To achieve this aim in a decentralized fashion, we treat the coupling term $\subss v i=\sum_{j\in\NN_i}\bar A_{ij}\subss e j$ as a disturbance for the error dynamics
          \begin{equation}
            \label{eq:erroridyntube}
            \subss\pe i=\bar{A}_{ii}\subss e i+\subss v i
          \end{equation}
          and then confine the error into an RPI set $\Sset_i\subseteq\Eset_i$ for
          \eqref{eq:erroridyntube}
          and $\subss v i\in\Vset_i=\bigoplus_{j\in\NN_i}\bar A_{ij}\Eset_j$. The main result, that will also enable PnP design of LSEs, is given in the next proposition.
          \begin{prop}
            \label{prop:main}
            Let $\Wset=\{0\}$. If, for given matrices $L_{ij}$ and parameters $\delta_{ij}$, $i,j\in\MM$, the following conditions are fulfilled
            \begin{subequations}
              \label{eq:pseudoinequalities}
              \begin{align}
                \label{eq:AiiSchur}&\bar A_{ii} \mbox{ is Schur, }\forall i\in\MM\\
                \label{eq:betapseudo}&\beta_i =\sum_{j\in\NN_i}\sum_{k=0}^{\infty}\norme{\HH_i\bar A_{ii}^k\bar A_{ij}\HH_j^\flat}{\infty}<1,~\forall i\in\MM
              \end{align}
            \end{subequations}
            then
            \begin{enumerate}[(I)]
            \item $\mbf{\bar A}$ is Schur;
            \item $\forall i\in\MM$ there exists an RPI $\Sset_i\subseteq\Eset_i$ for dynamics \eqref{eq:erroridyntube}, such that $\Sset=\prod_{i\in\MM}\Sset_i$ is a positively invariant set for system \eqref{eq:errordyn}.
            \end{enumerate}
          \end{prop}
          \begin{proof}
            The proof is given in the Appendix \ref{sec:proofmain}.
          \end{proof}
          % \begin{rmk}
          %   If the error dynamics \eqref{eq:errordyn} is decoupled, condition \eqref{eq:betapseudo} is always fulfilled and then the nominally convergence of the state estimator is guaranteed by condition \eqref{eq:AiiSchur}.
          % \end{rmk}
          % \begin{rmk}
          Some comments are in order. The conditions in Proposition \ref{prop:main} guarantee that if $ \subss e i (0)\in\Sset_i$, $\forall i\in\MM$, then \eqref{eq:conv} and \eqref{eq:bounderror} hold. Condition \eqref{eq:betapseudo}, that stems from the small gain theorem for networks \cite{Dashkovskiy2007}, implies that the coupling between subsystems must be sufficiently small. In particular, if subsystems are decoupled, \eqref{eq:betapseudo} is always fulfilled and nominal convergence of the state estimator is guaranteed by condition \eqref{eq:AiiSchur} only.
          % \begin{equation}
          %   \begin{aligned}
          %     &\norme{\mbf e(t)}{}\rightarrow 0\mbox{ as }t\rightarrow\infty \\
          %     &\mbf e (t)\in\Eset,~\forall t\geq 0
          %   \end{aligned}
          % \end{equation}
          % i.e. the state estimator is globally and nominally convergent.
          % \end{rmk}
          \begin{rmk}
            \label{rmk_trans_info}
            We highlight that, for a given $i\in\MM$, the quantity $\beta_i$ in \eqref{eq:pseudoinequalities} depends only upon local fixed parameters $\{A_{ii},C_{i},\HH_i\}$, neighbors' fixed parameters $\{A_{ij},C_j,\HH_j\}_{j\in\NN_i}$ and local tunable parameters $\{L_{ii},\{L_{ij},\delta_{ij}\}_{j\in\NN_i}\}$ but not on neighbors' tunable parameters. This implies that the choice of $\{L_{ii},\{L_{ij},\delta_{ij}\}_{j\in\NN_i}\}$ does not influence the choice of $\{L_{jj},\{L_{jk},\delta_{jk}\}_{k\in\NN_j}\}$, for $i\neq j$. \hfill{$\square$}
          \end{rmk}
          
          When system \eqref{eq:model} is affected by disturbances, i.e. $\Wset\neq\{0\}$, we can still use \eqref{eq:pseudoinequalities} for guaranteeing the stability of matrix $\mbf{\bar A}$, but we need an additional condition in order to guarantee the existence of an RPI set $\Sset_i\subseteq\Eset_i$ for the error dynamics
          \begin{equation}
            \label{eq:erroridyntubedist}
            \subss\pe i=\bar{A}_{ii}\subss e i+\subss {\tilde v} i
          \end{equation}
          where the disturbance $\subss{\tilde v} i$ verifies
          \begin{equation}
            \label{eq:Vsetdist}
            \subss{\tilde v}i=\subss v i+D_i \subss w i\in \tilde \Vset_i=(\bigoplus_{j\in\NN_i}\bar A_{ij}\Eset_j \oplus D_i \Wset_i).
          \end{equation}
          Since $\tilde\Vset_i$ is a zonotope, it can be written as $\tilde \Vset_i=\{\subss{\tilde v}{i}\in\Rset^{\tilde n_i}|\subss{\tilde v} i=\Psi_i \tilde d_i,\mbox{ }\norme{\tilde d_i}{\infty}\leq 1 \}$ where $\tilde n_i=\sum_{j\in\NN_i} n_j+r_i$ $\Psi_i = \matr{ \bar A_{ij_1}\Xi_{j_1} & \ldots & \bar A_{ij_z}\Xi_{j_z} & D_i\Delta_i}$, $j_1,\ldots,j_z\in\NN_i$.
          \begin{prop}
            \label{prop:maindist}
            For given matrices $L_{ij}$ and parameters $\delta_{ij}$, $i,j\in\MM$, if conditions \eqref{eq:pseudoinequalities} hold and
            \begin{equation}
              \label{eq:pseudoinequalitiesdist}
              \gamma_i =\sum_{k=0}^{\infty}\norme{\HH_i\bar A_{ii}^k\Psi_i}{\infty}<1,~\forall i\in\MM
            \end{equation}
            then, there exists an RPI set $\Sset_i\subseteq\Eset_i$ for \eqref{eq:erroridyntubedist}, such that $\Sset=\prod_{i\in\MM}\Sset_i$ is an RPI set for system \eqref{eq:errordyn}.
          \end{prop}
          \begin{proof}
            The proof is given in the Appendix \ref{sec:proofmainDist}.
          \end{proof}
          \begin{rmk}
            We note that if the subsystems are decoupled, then condition \eqref{eq:pseudoinequalitiesdist} implies that there exists an mRPI $\underline\Sset_i\subseteq \Eset_i$ for the local error dynamics \eqref{eq:erroridyntubedist}. Moreover, when subsystems are coupled and $\Wset_i=\{0\}$, if $\beta_i<1$ then $\gamma_i<1$. Indeed, $\Wset_i=\{0\}$ implies that $\Delta_i=0$ and, as shown in the proof of Proposition 1, it holds $\sum_{k=0}^{\infty}\sum_{j\in\NN_i}\norme{\HH_i\bar A_{ii}^k\bar A_{ij}\Xi_j^\flat}{\infty}\leq\sum_{k=0}^{\infty}\sum_{j\in\NN_i}\norme{\HH_i\bar A_{ii}^k\bar A_{ij}\HH_j^\flat}{\infty}$. Finally, the pieces of information needed for computing scalars $\gamma_i$ are the same needed for computing scalars $\beta_i$ (see Remark~\ref{rmk_trans_info}). \hfill{$\square$}
          \end{rmk}
          From results in Proposition \ref{prop:main} and \ref{prop:maindist}, Problem \ref{prob:estimator_properties} can be decomposed into the following independent design problems for $i\in\MM$.
          \subsubsection*{Problem $\PP_i$}
               \label{sec:design_Estimator_para}
               Check if there exist $L_{ii}$ and $\{L_{ij}\}_{j\in\NN_i}$ such that $\bar A_{ii}$ is Schur, $\beta_i<1$ and $\gamma_i<1$.\endproof

          \begin{rmk}
            \label{rmk:Esetchoice}
            As shown in \cite{Kolmanovsky1998}, a necessary condition for the existence of RPI sets $\Sset_i$ for \eqref{eq:erroridyntubedist} is that
            \begin{equation}
              \label{eq:bigE}
              \Eset_i\subseteq \tilde \Vset_i, ~\forall i\in\MM
            \end{equation}
            where $\tilde\Vset_i$ depend upon sets $\Eset_j$,
            $j\in\NN_i$, see \eqref{eq:Vsetdist}. 
            In our approach, sets $\Eset_i$ are assigned \textit{a priori} on the basis, e.g. of application-dependent constraints. Therefore we implicitly assume conditions \eqref{eq:bigE} are verified. However, if subsystems are added sequentially to an existing plant and LSEs are designed with the PnP procedure described in Section~\ref{sec:plugplay}, conditions \eqref{eq:bigE} are automatically checked and, if violated, they prevent from plugging-in subsystem $\subss \Sigma i$. We also highlight that when sets $\Eset_i$ can be arbitrarily chosen, centralized methods for fulfilling conditions \eqref{eq:bigE} exist \cite{Farina2011b}.
          \end{rmk}
          
          \subsection{Optimization-based synthesis of LSEs}
               \label{sec:computational}
               The procedure for solving problems $\PP_i$, $i\in\MM$ is summarized in Algorithm 1 that can be executed in parallel by each subsystem using local hardware.

               \begin{algorithm}
                 \caption{}
                 \label{alg:designdec}
                 \textbf{Input}: zonotopes $\Eset_i$, $\Wset_i$ and scalars $\delta_{ij}, \forall j\in\NN_i$.\\
                 \textbf{Output}: set $\Sset_i$ and state estimator $\subss{\tilde\Sigma}i$.\\
                 \begin{enumerate}
                 \item\label{alg:step1} if $\delta_{ij}=1$, compute the matrix $L_{ij}$, $\forall j\in\NN_i$ solving
                   \begin{equation}
                     \label{eq:Lijopt}
                     \min_{\substack{L_{ij}}} \norme{\HH_i\bar A_{ij}\HH_j^\flat}{p}
                   \end{equation}
                   where either $p=1$ or $p=F$.
                 \item\label{alg:step2} compute a matrix $L_{ii}$ such that $\beta_i<1$ and $\gamma_i<1$. If it does not exist \textbf{stop};
                 \item\label{alg:step3} compute the set $\Sset_i$.
                 \end{enumerate}
               \end{algorithm}

               In step (\ref{alg:step1}), if $\delta_{ij}=1$, the computation of matrices $L_{ij}$, $j\in\NN_i$ is required. Since the choice of $L_{ij}$ affects the coupling term $\bar A_{ij}=A_{ij}+\delta_{ij}L_{ij}C_j$, and hence the possibility of verifying inequalities \eqref{eq:pseudoinequalities} and \eqref{eq:pseudoinequalitiesdist}, we propose to reduce the magnitude of coupling by minimizing the magnitude of $\bar A_{ij}$ in \eqref{eq:Lijopt}, where $\HH_i$ and $\HH_j^\flat$ allow us to take into account the size of sets $\Eset_i$ and $\Eset_j$, respectively. More precisely, it can be shown that the term $\norme{\HH_i\bar A_{ij}\HH_j^\flat}{p}$ is a  measure of how much the coupling term  $\bar A_{ij}\subss e j$, $j\in\NN_i$  affects the fulfillment of the constraint $\subss e i\in\Eset_i$ (see Appendix \ref{sec:Lijoptnotes}). We highlight that the minimization of  $\| \HH_i\bar A_{ij}\HH_j^\flat\|_1$ in \eqref{eq:Lijopt} amounts to an LP problem and the minimization of $\| \HH_i\bar A_{ij}\HH_j^\flat\|_F$ can be recast into a Quadratic Programming (QP) problem. So far, the parameters $\delta_{ij}$ have been considered fixed. However, if in step (\ref{alg:step1}) one obtains $L_{ij}=0$ for some $j\in\NN_i$, it is impossible to reduce the magnitude of the coupling term $\bar A_{ij}$ and the knowledge of $\subss y j$ is useless for estimator $\subss{\tilde\Sigma} i$. This suggests to revise the choice of $\delta_{ij}$ and set $\delta_{ij}=0$.

               In step (\ref{alg:step2}), for the computation of matrix $L_{ii}$ we propose an automatic method in order to guarantee satisfaction of inequalities \eqref{eq:pseudoinequalities} and \eqref{eq:pseudoinequalitiesdist}. This procedure parallels the method proposed in \cite{Riverso2013c} for control design. Since in \eqref{eq:AiiSchur} we require the Schurness of matrix $\bar A_{ii}$, we need to guarantee that $L_{ii}$ stabilizes the pair $(A_{ii},C_i)$. In order to achieve this aim we design $L_{ii}$ as the dual LQ control gain associated to matrices $Q_i\geq0$ and $R_i>0$, i.e.
               \begin{equation*}
                 \label{eq:Lilq}
                 L_{ii}=(R_i+C_i\bar{P}_iC_i^T)^{-1}C_i\bar{P}_iA_{ii}^T
               \end{equation*}
               where $\bar{P}_i$ is the solution to the algebraic Riccati equation
               \begin{equation}
                 \label{eq:ricattiLii}
                 A_{ii}\bar{P}_iA_{ii}^T+Q_i-A_{ii}\bar{P}_iC_i^T(R_i+C_i\bar{P}_iC_i^T)^{-1}C_i\bar{P_i}A_{ii}^T=\bar{P}_i.
               \end{equation}
               We then solve the following nonlinear optimization problem
               \begin{subequations}
                 \label{eq:optimLi}
                 \begin{align}
                   \label{eq:costoptimLi}\min_{\substack{Q_i,~R_i}}&~\beta_i\\
                   &\label{eq:QRoptimLi}Q_i\geq 0,~R_i>0\\
                   &\label{eq:betaoptmLi}\beta_i<1\\
                   &\label{eq:thetaoptimLi}\gamma_i<1
                 \end{align}
               \end{subequations}
               where constraint \eqref{eq:thetaoptimLi} is needed only if $\Wset_i\neq\{0\}$. In order to simplify the optimization problem \eqref{eq:optimLi} one can assume $Q_i=\diag(q_{i,1},\ldots,q_{i,n_i})$, $R_i=\diag(r_{i,1},\ldots,r_{i,m_i})$ and replace the matrix inequalities in \eqref{eq:QRoptimLi} with the scalar inequalities $q_{i,k}\geq 0$, $k\in 1:n_i$ and $r_{i,k}> 0$, $k\in 1:m_i$. The feasibility of problem \eqref{eq:optimLi} guarantees that the estimator $\subss{\tilde\Sigma} i$ can be successfully designed. Note that if all matrices $L_{ij}$, $j\in\NN_i$ are such that $\bar A_{ij}=0$, the inequality \eqref{eq:betaoptmLi} is always fulfilled and, when $\Wset=\{0\}$, the optimization problem \eqref{eq:optimLi} is reduced to the solution of the algebraic Riccati equation \eqref{eq:ricattiLii}.

               In step (\ref{alg:step3}) of Algorithm \ref{alg:designdec} we need to compute a nonempty RPI set $\Sset_i\subseteq\Eset_i$ that, in view of Propositions \ref{prop:main} and \ref{prop:maindist}, exists if the optimization problem \eqref{eq:optimLi} is feasible. To this purpose, several algorithms can be used. For instance, \cite{Rakovic2005a} discusses the computation of $\epsilon$-outer approximation of the mRPI $\underline\Sset_i$. The MRPI set $\bar\Sset_i$ can be obtained using methods in \cite{Gilbert1991}. More recently, efficient procedures have been also proposed for computing polytopic \cite{Rakovic2010} or zonotopic \cite{Rakovic2005} RPI sets.

     \section{Plug-and-play operations}
          \label{sec:plugplay}
          Consider a plant composed by subsystems $\subss \Sigma i$, $i\in\MM$ equipped with local state estimators $\subss{\tilde\Sigma} i$, $i\in\MM$ produced by Algorithm \ref{alg:designdec}. In case subsystems are added or removed, we show how to preserve properties \eqref{eq:conv} and \eqref{eq:bounderror} by updating a limited number of existing LSEs. Note that plugging in and unplugging of subsystems are here considered as off-line operations, i.e. they do not lead to switching between different dynamics in real time.
          
          \subsection{Plugging in operation}
               \label{sec:plugin}
               We start considering the plugging in of subsystem $\subss{\Sigma}{M+1}$, characterized by parameters $A_{M+1,M+1}$, $C_{M+1}$, $\Eset_{M+1}$, $\Wset_{M+1}$, $\NN_{M+1}$ and $\{A_{M+1,j}\}_{j\in\NN_{M+1}}$. In particular, $\NN_{M+1}$ identifies the subsystems that will influence $\subss{\Sigma}{M+1}$ through matrices $\{A_{M+1,j}\}_{j\in\NN_{M+1}}$. Subsystems that will be influenced by $\subss \Sigma {M+1}$ are given by $\SSS_{M+1}$ where \[\SSS_i=\{j:i\in\NN_j\}\] is the set of children of subsystem $\subss \Sigma i$. For designing the LSE $\subss{\tilde\Sigma}{M+1}$ we execute Algorithm \ref{alg:designdec} that needs information only from subsystems $\subss{\Sigma}{j}$, $j\in\NN_{M+1}$. If Algorithm \ref{alg:designdec} stops before the last step, we declare that $\subss{\Sigma}{M+1}$ cannot be plugged in. Since sets $\NN_j$, $j\in\SSS_{M+1}$ have now one more element, previously obtained matrices $L_{jj}$, $j\in\SSS_{M+1}$ might give $\beta_i\geq 1$ or $\gamma_i\geq 1$. Indeed, quantities $\beta_i$ and $\gamma_i$ in \eqref{eq:pseudoinequalities} and \eqref{eq:pseudoinequalitiesdist} can only increase. Furthermore, the size of the set $\Sset_j$ increases and therefore the condition $\Sset_j\subseteq\Eset_j$ could be violated. This means that for each $j\in\SSS_{M+1}$ the LSE $\subss{\tilde\Sigma} j$ must be redesigned by running Algorithm \ref{alg:designdec}. Again, if Algorithm \ref{alg:designdec} stops before completion for some $j\in\SSS_{M+1}$, we declare that $\subss{\Sigma}{M+1}$ cannot be plugged in.\\
               Note that LSE redesign does not propagate further in the network, i.e. even without changing state estimators $\subss{\tilde\Sigma} i$, $i\notin\{M+1\}\bigcup\SSS_{M+1}$, properties  \eqref{eq:conv} and \eqref{eq:bounderror} are guaranteed for the new DSE.

          \subsection{Unplugging operation}
               \label{sec:unplug}
               We consider the unplugging of  system $\subss{\Sigma}{k}$, $k\in\MM$. Since for each $i\in\SSS_k$ the set $\NN_i$ contains one element less, one has that $\beta_i$ in \eqref{eq:pseudoinequalities} and $\gamma_i$ in \eqref{eq:pseudoinequalitiesdist} cannot increase. Furthermore, the set $\Sset_i^0$, chosen before the removal of system $\subss\Sigma k$, still verifies $\Sset_i^0\supseteq \tilde \Vset_i$ and therefore previously obtained optimizers for problem \eqref{eq:Lijopt} can still be used. This means that for each $i\in\SSS_k$ the LSE $\subss{\tilde\Sigma} i$ does not have to be redesigned. Moreover, since for each system $\subss\Sigma j$, $j\notin\{k\}\bigcup\SSS_k$, the set $\NN_j$ does not change, the redesign of the LSE $\subss{\tilde\Sigma} j$ is not required.\\
               In conclusion, the removal of system $\subss\Sigma k$ does not require the redesign of any LSE in order to guarantee \eqref{eq:conv} and \eqref{eq:bounderror}. However systems $\subss\Sigma i$ $i\in\SSS_k$ have one parent less and the redesign of LSEs $\subss{\tilde\Sigma} i$ through Algorithm \ref{alg:designdec} could improve the performance.

     \section{Example}
          \label{sec:example}
          We consider a system composed by $16$ masses coupled as in Figure \ref{fig:masses} where the four edges connected to a point correspond to springs and dampers arranged as in Figure \ref{fig:exampleCart2D}.
          \begin{figure}[htb]
            \centering
            \def\svgwidth{172pt}
            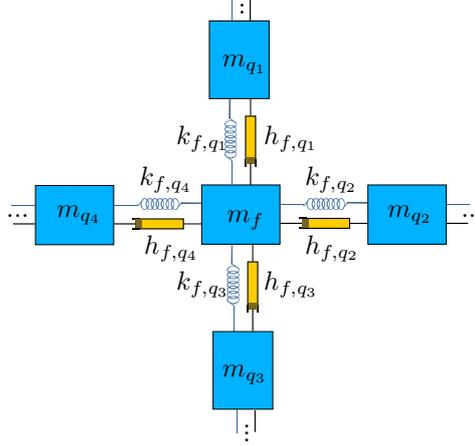
            \caption{Array of masses: details of interconnections.}
            \label{fig:exampleCart2D}
          \end{figure}
          Each mass $f\in1:16$ is an LTI system with state variables $\subss x f=(\subss x {f,1},\subss x {f,2},\subss x {f,3},\subss x {f,4})$ and input $\subss u f=(\subss u {f,1},\subss u {f,2})$, where $\subss x {f,1}$ and $\subss x {f,3}$ are the displacements of mass $f$ with respect to a given equilibrium position in the plane (equilibria lie on a regular grid), $\subss x {f,2}$ and  $\subss x {f,4}$ are the horizontal and vertical velocity of the mass $f$, respectively, and $100\subss u {f,1}$ (respectively $100\subss u {f,2}$) is the force applied to mass $f$ in the horizontal (respectively, vertical) direction. The values of $m_f$ have been extracted randomly in the interval $[5,10]$ while spring constants and damping coefficients are identical and equal to $0.5$. Each mass is equipped with local state estimation error constraints $\norme{\subss e {f,j}}{\infty}\leq 1$, $j=1,3$ and $\norme{\subss e {f,l}}{\infty}\leq 1.5$, $l=2,4$.

          A subsystem $\subss\Sigma i$, $i\in\MM=1:4$ is a group of four masses as in Figure \ref{fig:masses}. Therefore each subsystem has order $16$ and two neighbors. For each subsystem $\subss{\Sigma} i$ we have $8$ outputs that are the displacements of two masses and the velocities of the other two masses.
          \begin{figure}[!ht]
            \centering
            \includegraphics[scale=0.8]{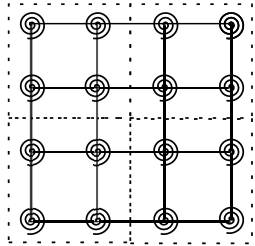}
            \caption{Position of the $16$ masses on the plane. Dashed lines define subsystems $\subss\Sigma i$, $i\in\MM=1:4$.}
            \label{fig:masses}
          \end{figure}
          We obtain models $\subss\Sigma i$ by discretizing continuous-time models with $0.2~$sec sampling time, using zero-order hold discretization for the local dynamics and treating $\subss x j,~j\in\NN_i$ as exogenous signals. We design an LSE $\subss{\tilde\Sigma} i$, $i\in\MM$ using Algorithm \ref{alg:designdec} and assuming matrices $Q_i$ and $R_i$ in \eqref{eq:optimLi} are diagonal. In Figure \ref{fig:nodist} we show a simulation where the initial state of each mass is $\subss x f(0)=0$, $f\in 1:16$ and the control inputs $\subss u {f,l}(k)=0.1\sin(k)$, $l\in 1:2$, have been used. We initialize each LSE in order to have $\subss e i\in\Sset_i$. Estimation results produced by LSEs that have been designed with $\delta_{ij}=0$, $j\in\NN_i$ are represented in Figures \ref{fig:sub1nodist} and \ref{fig:sub3nodist}. Results obtained by setting $\delta_{ij}=1$, $j\in\NN_i$ are shown in Figures~\ref{fig:sub2nodist} and \ref{fig:sub4nodist}. One can notice that in both cases, state estimation errors converge to zero and they are bounded at all times.
          \begin{figure}[!ht]
            \centering
            \subfigure[\label{fig:sub1nodist}State (dashed lines) and state estimation (continuous line) of the upper left mass in Figure \ref{fig:masses} at time instants $k=0:29$.]{\includegraphics[scale=0.22]{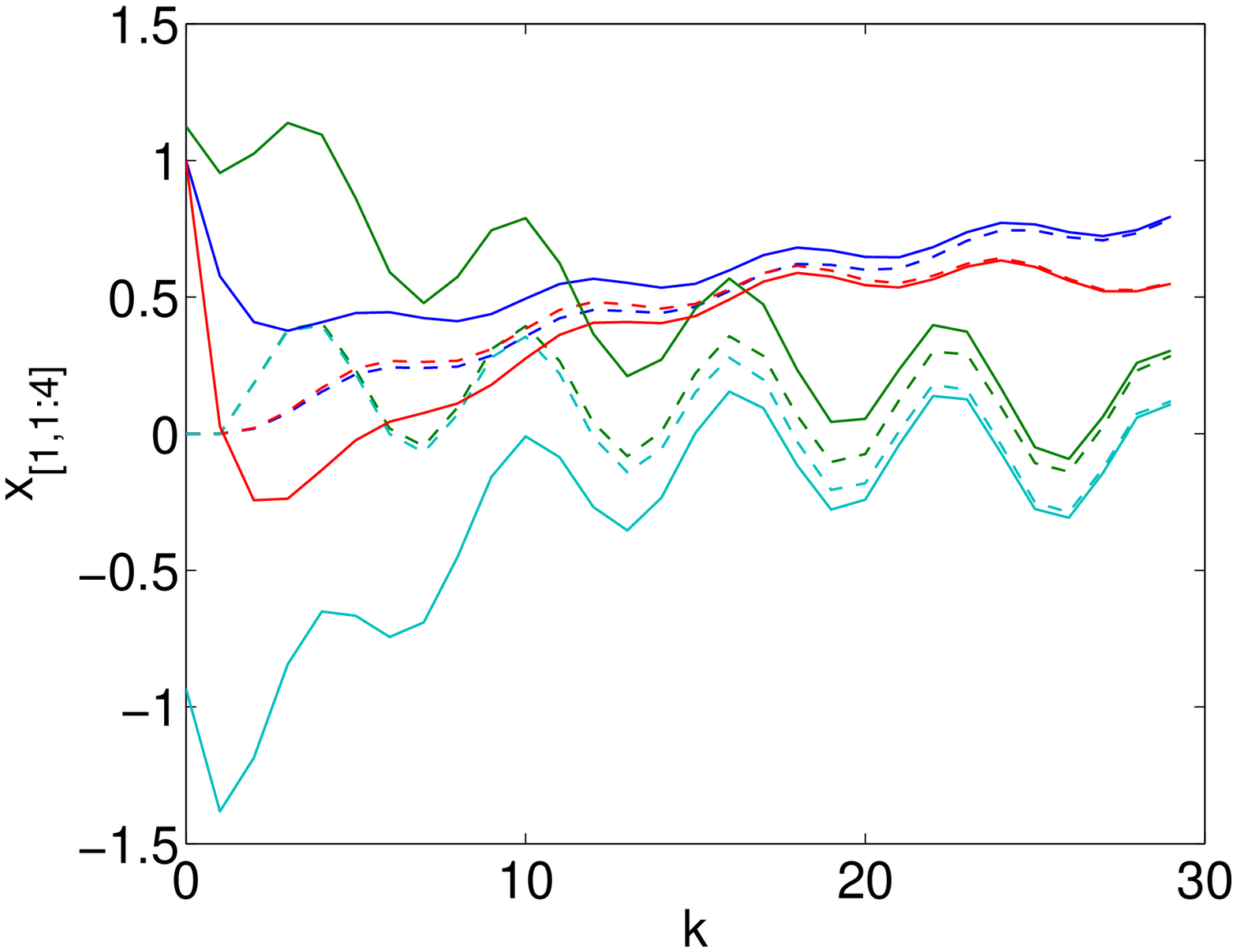}}~~
            \subfigure[\label{fig:sub2nodist}State (dashed lines) and state estimation (continuous line) of the upper left mass in Figure \ref{fig:masses} at time instants $k\in0:29$.]{\includegraphics[scale=0.22]{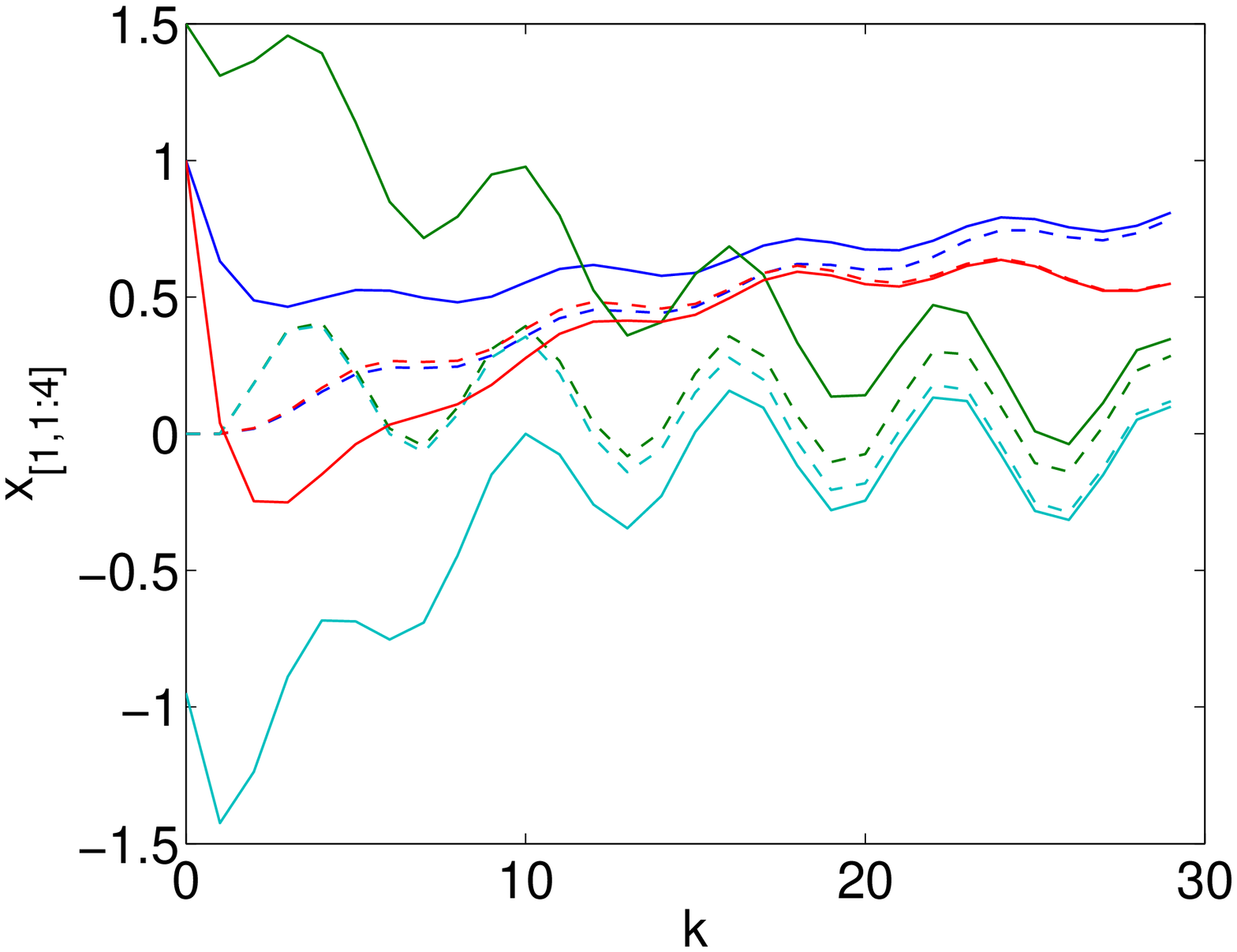}}\\
            \subfigure[\label{fig:sub3nodist}Estimation errors for all states at times $k\in0:99$.]{\includegraphics[scale=0.22]{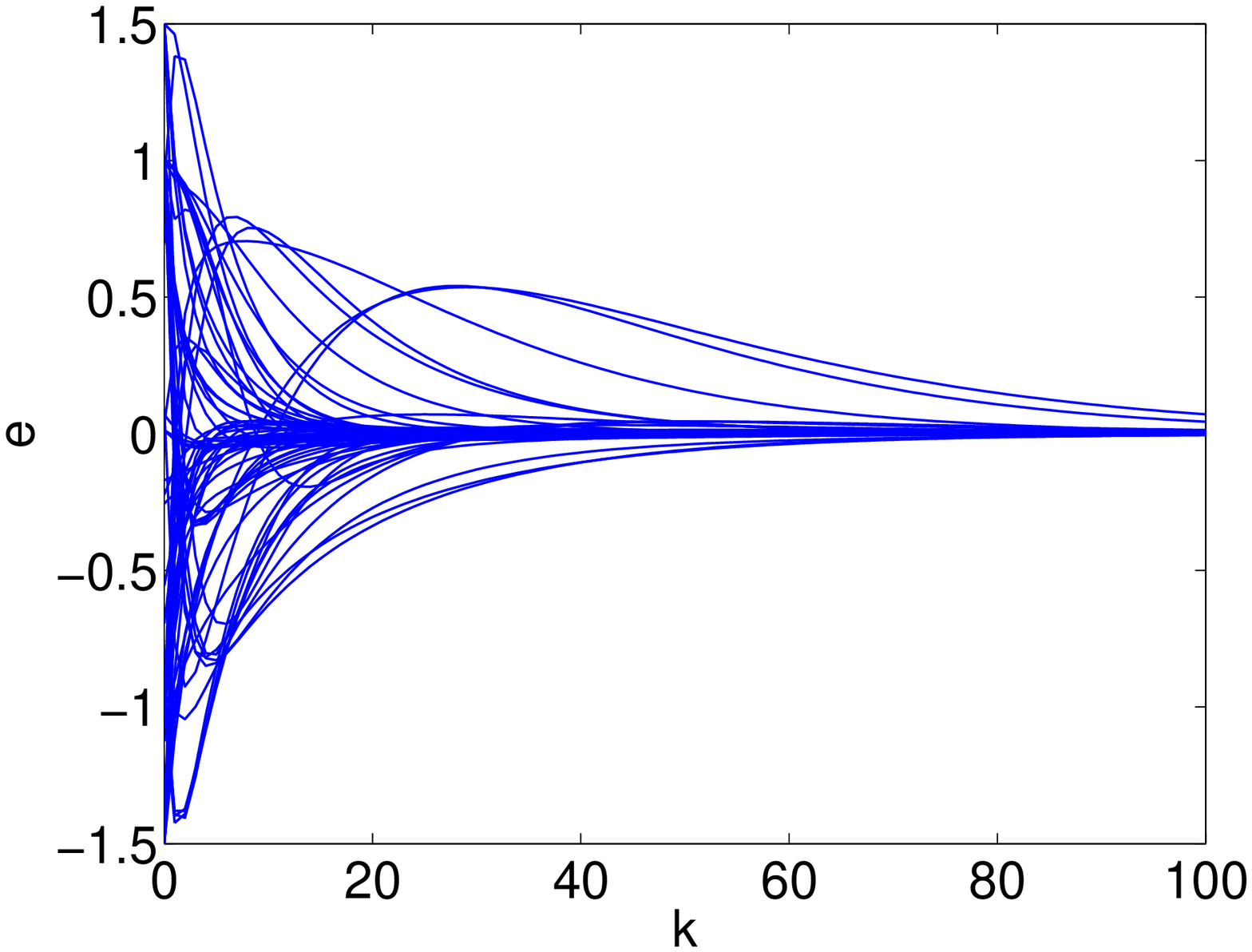}}~~
            \subfigure[\label{fig:sub4nodist}Estimation errors for all states at times $k\in0:99$.]{\includegraphics[scale=0.22]{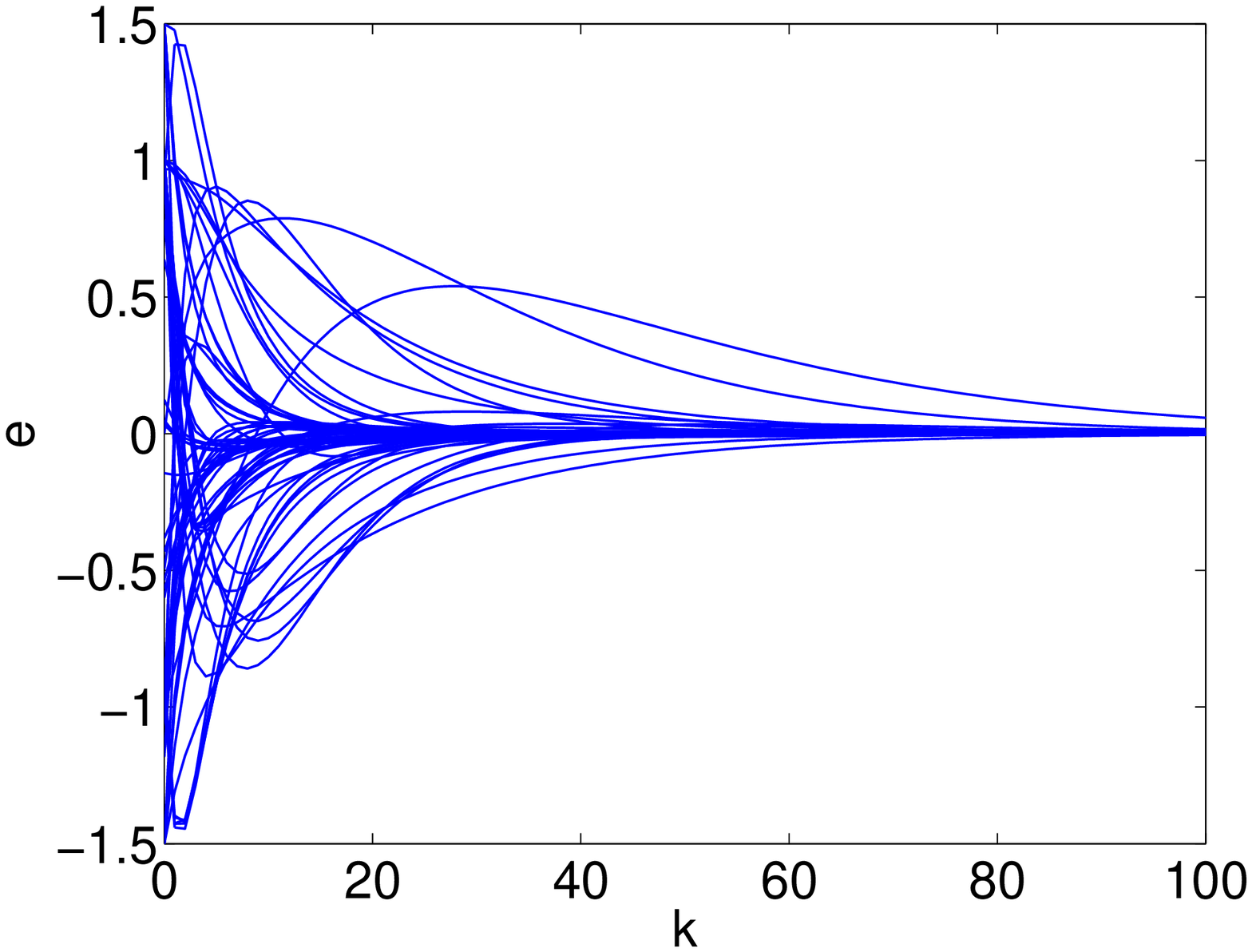}}
            \caption{State estimation results for LSEs  designed setting $\delta_{ij}=0$, $j\in\NN_i$ (panels \ref{fig:sub1nodist} and \ref{fig:sub3nodist}) and $\delta_{ij}=1$, $j\in\NN_i$ (panels \ref{fig:sub2nodist} and \ref{fig:sub4nodist}). In panels \ref{fig:sub1nodist} and \ref{fig:sub2nodist} the same color has been used for a state and its estimate: cyan and green lines denote velocities while blue and red lines denote positions.}
            \label{fig:nodist}
          \end{figure}
          
          In Figure \ref{fig:dist} we show a simulation where each state of subsystem $\subss\Sigma i$, $i\in 1:4$ is affected by a disturbance $\subss w i$ sampled from the uniform distribution in the set $\Wset_i=\{\subss w i\in\Rset:\abs{\subss w i}\leq 0.015\}$. This has been obtained setting $D_i=\One_{16}$.

          Figures \ref{fig:sub1dist} and \ref{fig:sub3dist} show results produced by LSEs designed with $\delta_{ij}=0$, $j\in\NN_i$ while Figures \ref{fig:sub2dist} and \ref{fig:sub4dist} show the results obtained for $\delta_{ij}=1$, $j\in\NN_i$. In both cases, errors fulfill the prescribed bounds but do not converge to zero because of the persistent disturbances $\subss w i$, $i\in 1:4$.
          \begin{figure}[!ht]
            \centering
            \subfigure[\label{fig:sub1dist}State (dashed lines) and state estimation (continuous line) of the upper left mass in Figure \ref{fig:masses} at time instants $k=0:29$.]{\includegraphics[scale=0.22]{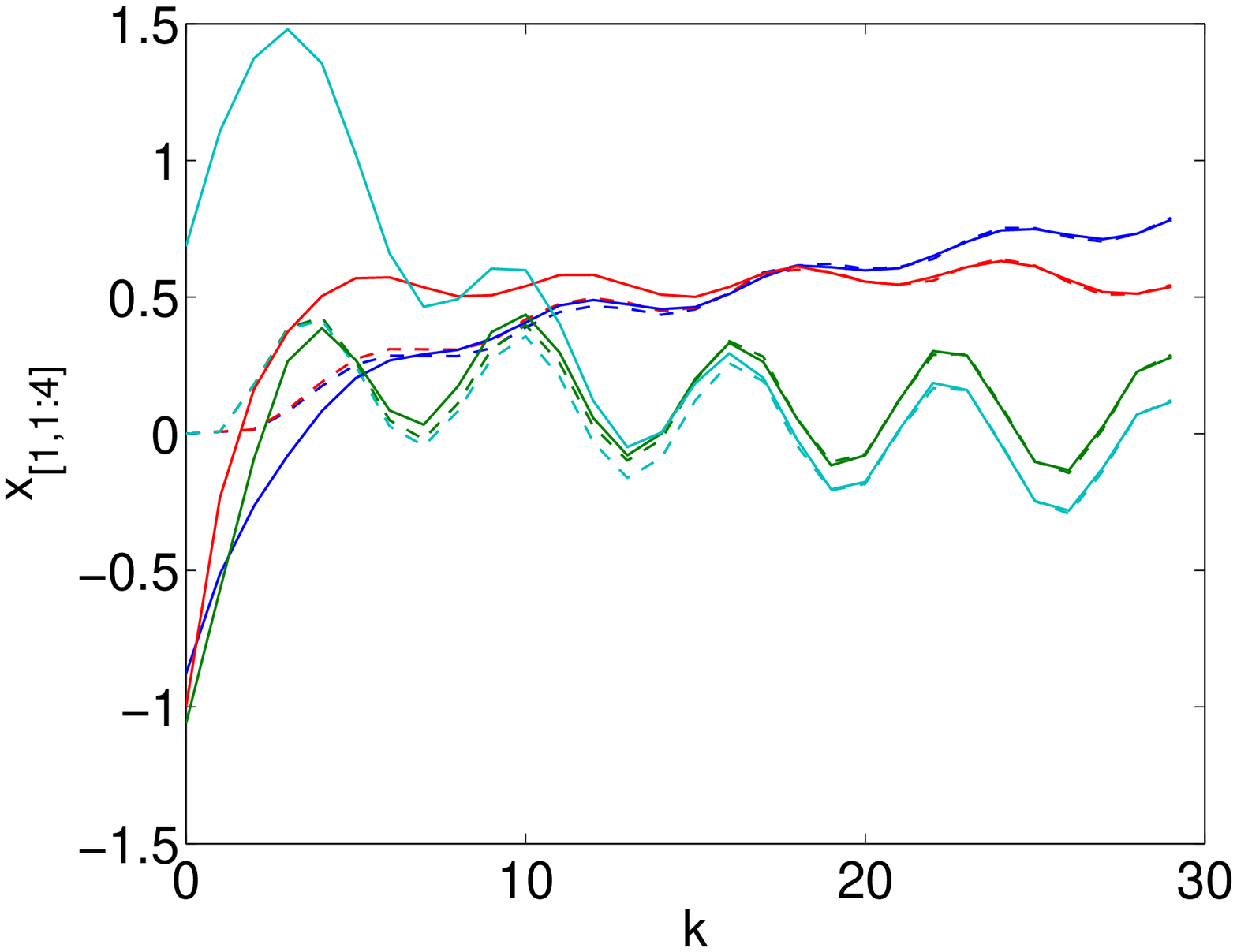}}~~
            \subfigure[\label{fig:sub2dist}State (dashed lines) and state estimation (continuous line) of the upper left mass in Figure \ref{fig:masses} at time instants $k=0:29$.]{\includegraphics[scale=0.22]{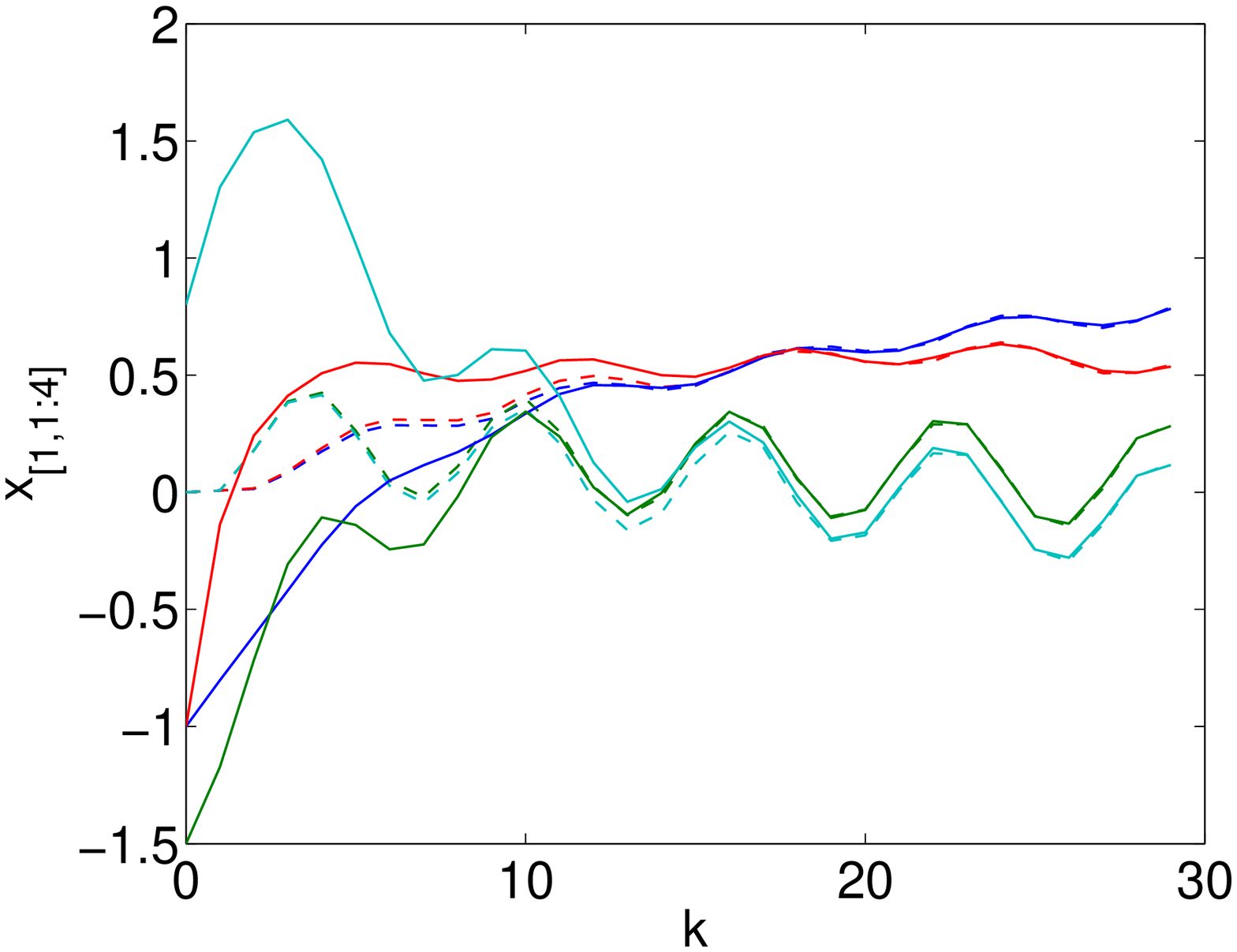}}\\
            \subfigure[\label{fig:sub3dist}Estimation errors for all states at times $k\in0:99$.]{\includegraphics[scale=0.22]{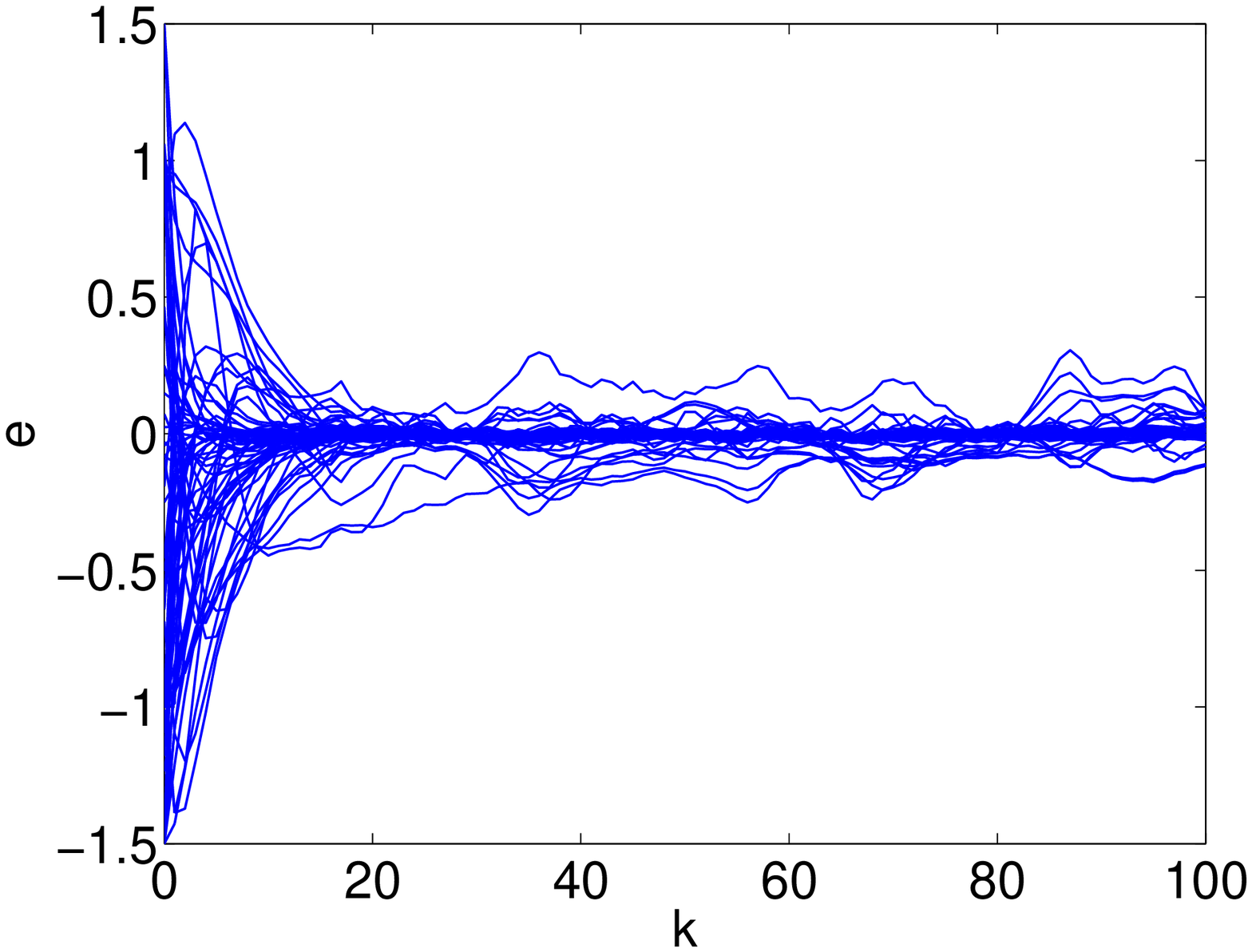}}~~
            \subfigure[\label{fig:sub4dist}Estimation errors for all states at times $k\in0:99$.]{\includegraphics[scale=0.22]{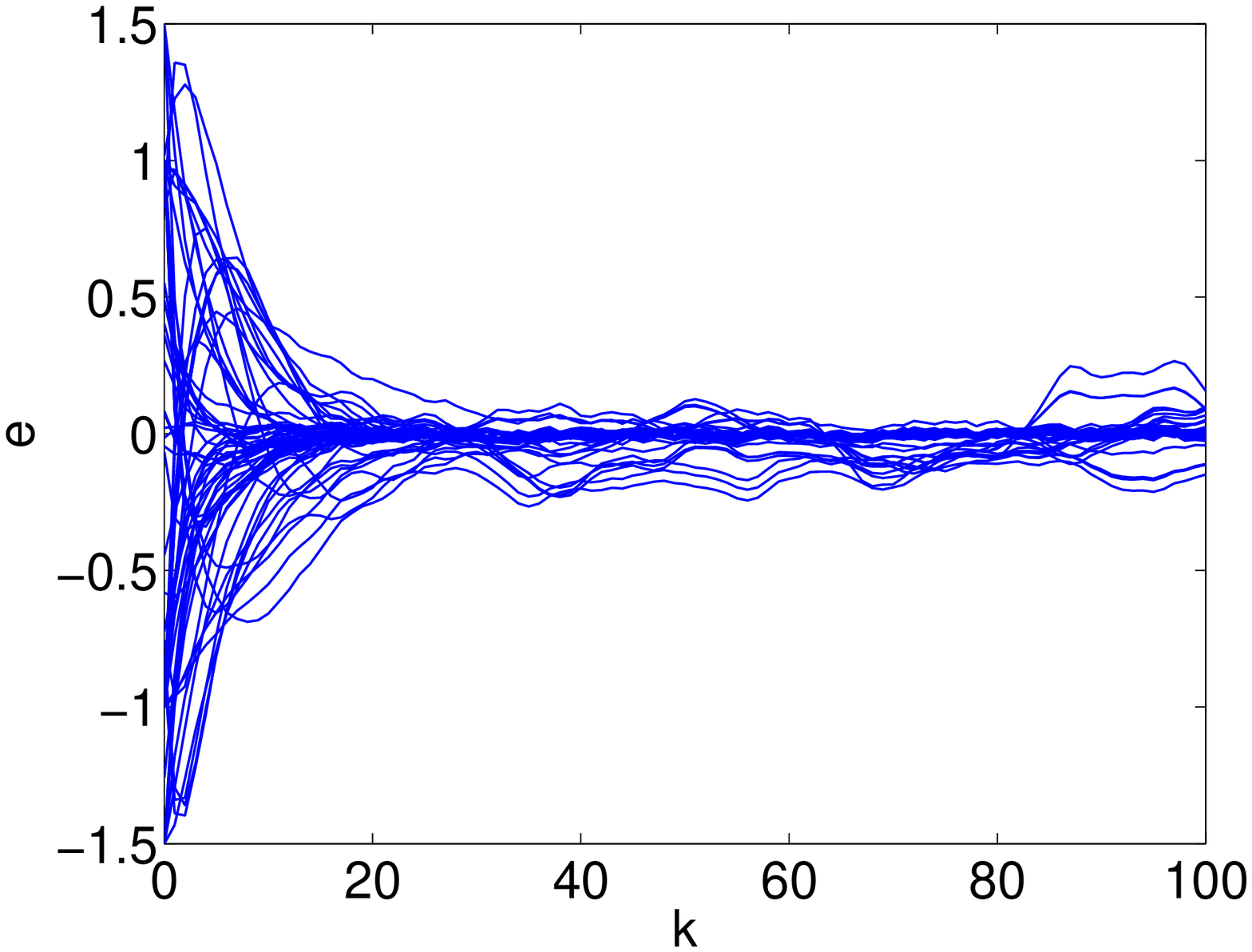}}
            \caption{State estimation results for LSEs designed setting $\delta_{ij}=0$, $j\in\NN_i$ (panels \ref{fig:sub1dist} and \ref{fig:sub3dist}) and $\delta_{ij}=1$, $j\in\NN_i$ (panels \ref{fig:sub2dist} and \ref{fig:sub4dist}). In panels \ref{fig:sub1dist} and \ref{fig:sub2dist} the same color has been used for a state and its estimate: cyan and green lines denote velocities while blue and red lines denote positions.}
            \label{fig:dist}
          \end{figure}
          
     \section{Conclusions}
          \label{sec:conclusions}
          We have proposed a novel DSE for large-scale linear perturbed systems, which guarantees that the estimation errors are bounded into prescribed sets and converge to zero in absence of disturbances. The algorithm is based on the partition of the overall system into subsystems with non-overlapping states. In particular, the design of LSEs can be carried out in a decentralized fashion by solving a suitable optimization problem where just information by parent nodes is required. This allows one to efficiently update the overall DSE when subsystems are plugged in and out. \\
          Future works include the design of output-feedback PnP schemes combining the state estimator proposed in this paper and the state-feedback PnP controller presented in \cite{Riverso2013c}.

     \section{Appendix}
          \subsection{Proof of Proposition \ref{prop:main}}
               \label{sec:proofmain}
               The proof uses arguments that are similar to the ones adopted for proving points (I) and (II) of Theorem 2 in \cite{Riverso2012a}.
               \subsubsection{Proof of (I)}
                    Define a matrix $\Metz$ such that its $ij$-th entry $\mu_{ij}$ is
                    $$
                    \begin{array}{lcl}
                      \mu_{ij}=-1&\text{if}&i=j\\
                      \mu_{ij}=\sum_{k=0}^{\infty}\norme{\HH_i\bar A_{ii}^k\bar A_{ij}\HH_j^\flat}{\infty}&\text{if}&i\neq j.
                    \end{array}
                    $$
                    Note that all the off-diagonal entries of matrix $\Metz$ are non-negative, i.e., $\Metz$ is Metzler~\cite{Farina2000}. We recall the following results.
                    \begin{lem}[see \cite{Mason2007}]
                      \label{lem:metzhur}
                      Let matrix $\Metz\in\Rset^{M\times M}$ be Metzler. Then $\Metz$ is Hurwitz if and only if there is a vector $\nu\in\Rset_+^M$ such that $\Metz\nu<\Zero_M$.
                    \end{lem}
                    \begin{lem}
                      \label{lem:metzschur}
                      Define the matrix $\Gamma=\Metz+\eye M$ where $\Metz\in\Rset^{M\times M}$, $\eye M$ is the $M\times M$ identity matrix and  $\Gamma$ is non negative. Then the Metzler matrix $\Metz$ is Hurwitz if and only if $\Gamma$ is Schur.
                    \end{lem}
                    The proof of Lemma~\ref{lem:metzschur} easily follows from Theorem 13 in \cite{Farina2000}.\\
                    Inequalities~\eqref{eq:pseudoinequalities} are equivalent to $\Metz\nu<\Zero_M$ where $\nu=\One_M$. Then, from Lemma \ref{lem:metzhur}, $\Metz$ is Hurwitz. From Lemma \ref{lem:metzschur}, \eqref{eq:pseudoinequalities} implies that matrix $\Gamma= \Metz+\eye M$ is Schur.\\
                    For dynamics \eqref{eq:erroridyntube}, we have
                    \begin{equation}
                      \label{eq:lagruzero}
                      \subss e i (t)=\bar A_{ii}^t\subss e i(0)+\sum_{k=0}^{t-1}\bar A_{ii}^k\sum_{j\in\NN_i}\bar A_{ij}\subss e j(t-k-1)
                    \end{equation}
                    In view of \eqref{eq:lagruzero} we can write
                    \begin{equation*}
                      \label{eq:lagruzeronorm}
                      \oneblock{
                        \norme{\HH_i\subss e i (t)}{\infty}
                                                            &\leq\norme{\HH_i\bar A_{ii}^t\HH_i^\flat}{\infty}\norme{\HH_i\subss e i(0)}{\infty}+\\ &+\sum_{j\in\mathcal{N}_i}\gamma_{ij}\max_{\substack{k\leq t}}\norme{\HH_j \subss e j(k)}{\infty}.
                      }
                    \end{equation*}
                    where $\gamma_{ij}$ are the entries of $\Gamma$. Denoting $\subss{\tilde e} i=\HH_i \subss e i$, we can collectively define $\tilde{\mbf e}=\tilde{\mbf\HH}\mbf e$, where $\tilde{\mbf\HH}=\diag(\HH_1,\dots,\HH_M)$. From the definition of sets $\Eset_i$, we have rank$(\tilde{\mbf\HH})=n$. We define the system
                    \begin{align}
                      \label{eq:exp_sys}
                      \mbf{\tilde e}^+=\tilde{\mbf\bar A}\tilde{\mbf e}
                    \end{align}
                    where $\tilde{\mbf \bar A}=\tilde{\mbf\HH}\mbf{\bar A}\tilde{\mbf\HH}^{\flat}$. In order to analyze the stability of the origin of \eqref{eq:exp_sys}, we use the small gain theorem for networks in \cite{Dashkovskiy2007}. In view of Corollary 16 in \cite{Dashkovskiy2007}, the overall system \eqref{eq:exp_sys} is asymptotically stable if the gain matrix $\Gamma$ is Schur and, as shown above, this property is implied by \eqref{eq:pseudoinequalities}. Moreover, system~\eqref{eq:exp_sys} is an expansion of the original system (see Chapter 3.4 in \cite{Lunze1992}). In view of the inclusion principle \cite{Stankovic2004}, the asymptotic stability of \eqref{eq:exp_sys} implies the asymptotic stability of the original system.

               \subsubsection{Proof of (II)}
                    \label{sec:proof1part2}
                    First note that, for $i\in\MM$, since $\Eset_i$ is a zonotope, $\norme{h_{i,\tau}^T\Xi_i}{\infty}=1$ for all $\tau\in1:\bar\tau_i$ and therefore $\norme{\HH_i\Xi_i}{\infty}=1$. This implies that $\norme{h_{i,\tau}^T\bar A_{ii}^k\bar A_{ij}\Xi_j}{\infty} \leq\norme{h_{i,\tau}^T\bar A_{ii}^k\bar A_{ij}\HH_{j}^{\flat}}{\infty}\norme{\HH_{j}\Xi_j}{\infty}=\norme{h_{i,\tau}^T\bar A_{ii}^k\bar A_{ij}\HH_{j}^{\flat}}{\infty}\leq\norme{\HH_{i}\bar A_{ii}^k\bar A_{ij}\HH_{j}^{\flat}}{\infty}$.

                    Therefore, from~\eqref{eq:betapseudo}, for all $\tau\in 1:\bar{\tau}_i$ it holds
                    \begin{equation}
                      \label{eq:errorcondition}
                      \sum_{k=0}^\infty\sum_{j\in\NN_i}\norme{h_{i,\tau}^T\bar A_{ii}^k\bar A_{ij}\Xi_j}{\infty}\leq\sum_{k=0}^\infty\sum_{j\in\NN_i}\norme{\HH_{i}\bar A_{ii}^k\bar A_{ij}\HH_{j}^{\flat}}{\infty}<1
                    \end{equation}
                    The next aim is to prove that there exists an RPI $\Sset_i\subseteq\Eset_i$ for the dynamics \eqref{eq:erroridyntube}, in particular we define $\Sset_i$ as an outer approximation of the mRPI $\underline\Sset_i$ and we prove that the outer approximation always exists.\\
                    The mRPI for \eqref{eq:erroridyntube} is given by \cite{Rakovic2005a}
                    \begin{equation}
                      \label{eq:mRPI}
                      \underline\Sset_i=\bigoplus_{k=0}^\infty \bar A_{ii}^k\bigoplus_{j\in\NN_i}\bar A_{ij}\Eset_j.
                    \end{equation}
                    From \cite{Rakovic2005a}, for given $\epsilon_i>0$ there exist $\alpha_i\in\Rset$ and $s_i\in\Nset_+$ such that the set
                    \begin{equation}
                      \label{eq:approxmRPI}
                      \Sset_i(\epsilon_i)=(1-\epsilon_i)^{-1}\bigoplus_{k=0}^{s_i-1} \bar A_{ii}^k\bigoplus_{j\in\NN_i}\bar A_{ij}\Eset_j
                    \end{equation}
                    is an $\epsilon_i-$outer approximation of the mRPI $\underline\Sset_i$.\\
                    Using arguments from Section 3 of \cite{Kolmanovsky1998}, we can then guarantee that $\Sset_i(\epsilon_i)\subseteq\Eset_i$. In fact for all $\tau\in 1:\bar{\tau}_i$
                    \begin{equation}
                      \label{eq:supSeps}
                      \sup_{\substack{\subss{s}{i}\in\Sset_i(\epsilon_i)}}~h_{i,\tau}^T\subss{s}{i}\leq 1.
                    \end{equation}
                    Using \eqref{eq:mRPI}, the inequalities \eqref{eq:supSeps} are verified if
                    \begin{equation}
                      \label{eq:supmSeps}
                      \sup_{\substack{\{\subss{e}{j}(k)\in \Eset_j\}_{j\in\NN_i}^{k=0,\ldots,\infty}\\\sigma_i\in\ball{\epsilon_i}(0)}}~z^x_{i,\tau}(\{\subss{e}{j}(k)\}_{j\in\NN_i}^{k=0,\ldots,\infty})+\norme{h_{i,\tau}^T\sigma_i}{\infty}\leq 1
                    \end{equation}
                    where $z^x_{i,\tau}(\cdot)=h_{i,\tau}^T\sum_{k=0}^\infty \bar A_{ii}^k\sum_{j\in\NN_i}\bar A_{ij}\subss{e}{j}(k)$.\\
                    Since $\norme{h_{i,r}^T\sigma_i}{\infty}\leq\norme{h_{i,r}^T}{\infty}\epsilon_i$, conditions \eqref{eq:supmSeps} are satisfied if
                    \begin{equation}
                      \label{eq:supmSwB}
                      \sup_{\substack{\{\subss{e}{j}(k)\in \Sset_j\}_{j\in\NN_i}^{k=0,\ldots,\infty}}}~~z^x_{i,\tau}(\{\subss{e}{j}(k)\}_{j\in\NN_i}^{k=0,\ldots,\infty})\leq 1-\norme{h_{i,\tau}^T}{\infty}\epsilon_i.
                    \end{equation}
                    Using \eqref{eq:setspolyE}, we can rewrite \eqref{eq:supmSwB} as
                    \begin{equation}
                      \label{eq:supmSwB2}
                      \sup_{\substack{\{\norme{d_j(k)}{\infty}\leq 1\}_{j\in\NN_i}^{k=0,\ldots,\infty}}}~~z^d_{i,\tau}(\{d_j(k)\}_{j\in\NN_i}^{k=0,\ldots,\infty})\leq 1-\norme{h_{i,r}^T}{\infty}\epsilon_i
                    \end{equation}
                    where $z^d_{i,\tau}(\cdot)=h_{i,\tau}^T(\sum_{k=0}^\infty \bar A_{ii}^k\sum_{j\in\NN_i}\bar A_{ij}\Xi_jd_j(k))$.\\
                    The inequalities \eqref{eq:supmSwB2} are satisfied if
                    \begin{equation}
                      \label{eq:supnorm}
                      \sum_{k=0}^\infty\sum_{j\in\NN_i}\norme{h_{i,\tau}^T\bar A_{ii}^k\bar A_{ij}\Xi_j}{\infty}\leq 1-\norme{h_{i,\tau}^T}{\infty}\epsilon_i
                    \end{equation}
                    for all $\tau\in1:\bar\tau_i$.\\
                    In view of~\eqref{eq:errorcondition}, there exists a sufficiently small $\epsilon_i>0$ satisfying~\eqref{eq:supnorm}. Hence we proved that $\forall i\in\MM$ there exists an RPI $\Sset_i\subseteq\Eset_i$ for dynamics \eqref{eq:erroridyntube}. Moreover if we define $\Sset=\prod_{i\in\MM}\Sset_i$, the set $\Sset$ is an invariant set for system \eqref{eq:errordyn} equipped with constraints $\Eset$ and $\Wset=\{0\}$.

          \subsection{Proof of Proposition \ref{prop:maindist}}
               \label{sec:proofmainDist}
               In the following we use similar arguments of Proof of Proposition \ref{prop:main} (see Section \ref{sec:proof1part2}) to prove that there exists an RPI $\Sset_i\subseteq\Eset_i$ for the dynamics \eqref{eq:erroridyntubedist}, in particular we define $\Sset_i$ as an outer approximation of the mRPI $\underline\Sset_i$ and we prove that the outer approximation always exists.\\
               The mRPI for \eqref{eq:erroridyntubedist} is given by \cite{Rakovic2005a}
               \begin{equation}
                 \label{eq:mRPIdist}
                 \underline\Sset_i=\bigoplus_{k=0}^\infty \bar A_{ii}^k \left(\bigoplus_{j\in\NN_i}\bar A_{ij}\Eset_j \oplus D_i\Wset_i\right) = \bigoplus_{k=0}^\infty \bar A_{ii}^k\tilde\Vset_i.
               \end{equation}
               From \cite{Rakovic2005a}, for given $\epsilon_i>0$ there exist $\alpha_i\in\Rset$ and $s_i\in\Nset_+$ such that the set
               \begin{equation}
                 \label{eq:approxmRPIdist}
                 \Sset_i(\epsilon_i)=(1-\epsilon_i)^{-1}\bigoplus_{k=0}^{s_i-1} \bar A_{ii}^k\tilde\Vset_i
               \end{equation}
               is an $\epsilon_i-$outer approximation of the mRPI $\underline\Sset_i$.\\
               Using arguments from Section 3 of \cite{Kolmanovsky1998}, we can then guarantee that $\Sset_i(\epsilon_i)\subseteq\Eset_i$. In fact for all $\tau\in 1:\bar{\tau}_i$
                    \begin{equation}
                      \label{eq:supSepsdist}
                      \sup_{\substack{\subss{s}{i}\in\Sset_i(\epsilon_i)}}~h_{i,\tau}^T\subss{s}{i}\leq 1.
                    \end{equation}
                    Using \eqref{eq:mRPIdist}, the inequalities \eqref{eq:supSepsdist} are verified if
                    \begin{equation}
                      \label{eq:supmSepsdist}
                      \sup_{\substack{\sigma_i\in\ball{\epsilon_i}(0)\\\subss {\tilde v} i\in\Vset_i}}~z^x_{i,\tau}(\{\subss{\tilde v}{i}(k)\}^{k=0,\ldots,\infty})+\norme{h_{i,\tau}^T\sigma_i}{\infty}\leq 1
                    \end{equation}
                    where $z^x_{i,\tau}(\cdot)=h_{i,\tau}^T\sum_{k=0}^\infty \bar A_{ii}^k\subss{\tilde v} i$.\\
                    Since $\norme{h_{i,r}^T\sigma_i}{\infty}\leq\norme{h_{i,r}^T}{\infty}\epsilon_i$, conditions \eqref{eq:supmSeps} are satisfied if
                    \begin{equation}
                      \label{eq:supmSwBdist}
                      \sup_{\substack{\subss{\tilde v} i\in\tilde\Vset_i}}~~z^x_{i,\tau}(\{\subss{\tilde v}{i}(k)\}^{k=0,\ldots,\infty})\leq 1-\norme{h_{i,\tau}^T}{\infty}\epsilon_i.
                    \end{equation}
                    Using \eqref{eq:setspolyE} and \eqref{eq:setspolyW}, we can rewrite \eqref{eq:supmSwBdist} as
                    \begin{equation}
                      \label{eq:supmSwB2dist}
                      \sup_{\substack{\{\norme{\tilde d_i(k)}{\infty}\leq 1\}^{k=0,\ldots,\infty} }}~~z^d_{i,\tau}(\{\tilde d_i(k)\}^{k=0,\ldots,\infty})\leq 1-\norme{h_{i,r}^T}{\infty}\epsilon_i
                    \end{equation}
                    where $z^d_{i,\tau}(\cdot)=h_{i,\tau}^T(\sum_{k=0}^\infty \bar A_{ii}^k\Psi_i\tilde d_i(k))$.\\
                    The inequalities \eqref{eq:supmSwB2dist} are satisfied if
                    \begin{equation}
                      \label{eq:supnormdist}
                      \sum_{k=0}^\infty\norme{h_{i,\tau}^T\bar A_{ii}^k\Psi_i}{\infty}\leq 1-\norme{h_{i,\tau}^T}{\infty}\epsilon_i
                    \end{equation}
                    for all $\tau\in1:\bar\tau_i$.\\
                    We proved that $\forall i\in\MM$ there exists an RPI $\Sset_i\subseteq\Eset_i$ for dynamics \eqref{eq:erroridyntubedist}. Moreover if we define $\Sset=\prod_{i\in\MM}\Sset_i$, the set $\Sset$ is an RPI invariant set for system \eqref{eq:errordyn} equipped with constraints $\Eset$ and $\Wset\neq\{0\}$.

          \subsection{Notes on the optimization problem \eqref{eq:Lijopt}}
               \label{sec:Lijoptnotes}
               In order to fulfill condition \eqref{eq:betapseudo}, we need to guarantee at least that
               $$
               \bar A_{ij}\Eset_j\subseteq\Eset_i
               $$
               hence
               $$
               \HH_i \bar A_{ij}\subss e j\leq 1, \forall \subss e j\in\Eset_j.
               $$
               In order to minimize the effect of coupling terms $\bar
               A_{ij}$, from \eqref{eq:setspolyE} we can solve the following optimization problem.
               \begin{equation}
                 \label{eq:minmaxAij1}
                 \eta_{ij} = \min_{L_{ij}}\max_{\substack{\subss e j=\Xi_jd_j\\\norme{d_j}{\infty}\leq 1}}\norme{\HH_i\bar A_{ij}\subss e j}{p}.
               \end{equation}
               where $p=1$ or $p=F$. Using arguments similar to the
               ones adopted in the proof of Proposition \ref{prop:main}, from \eqref{eq:minmaxAij1} we obtain
               \begin{equation}
                 \label{eq:minmaxAij2}
                 \begin{aligned}
                   \eta_{ij} &\leq \min_{L_{ij}}\max_{\norme{d_j}{\infty}\leq 1}\norme{\HH_i\bar A_{ij}\Xi_jd_j}{p}\\
                               &\leq \min_{L_{ij}}\max_{\norme{d_j}{\infty}\leq 1}\norme{\HH_i\bar A_{ij}\HH_j^\flat}{p}\norme{\HH_j\Xi_jd_j}{p}\\
                               &\leq \min_{L_{ij}}\max_{\norme{d_j}{\infty}\leq 1}\norme{\HH_i\bar A_{ij}\HH_j^\flat}{p}\norme{\HH_j\Xi_j}{p}\norme{d_j}{p}
                 \end{aligned}
               \end{equation}
               Irrespectively of $p$, there exist constants $c_{1,p}>0$ and $c_{2,p}>0$ such that
               $$
               \norme{\HH_j\Xi_j}{p}\leq c_{1,p}\norme{\HH_j\Xi_j}{\infty}=c_{1,p}
               $$
               $$
               \max_{\norme{d_j}{\infty}\leq 1}\norme{d_j}{p}\leq\max_{\norme{d_j}{\infty}\leq 1} c_{2,p}\norme{d_j}{\infty}=c_{2,p}
               $$
               Therefore, we can conclude that
               \begin{equation*}
                 \eta_{ij}\leq c_{1,p}c_{2,p}\min_{L_{ij}}\norme{\HH_i\bar A_{ij}\HH_j^\flat}{p}
               \end{equation*}
               and this motivates the optimization problem \eqref{eq:Lijopt}.

     \bibliographystyle{IEEEtran}
     \bibliography{PnP_DeDi_Observer}

     % \addtolength{\textheight}{-3cm}

\end{document}